\documentclass[]{imsart}

\RequirePackage{amsthm,amsmath,amsfonts,amssymb}
\RequirePackage[authoryear]{natbib}

\startlocaldefs
\usepackage{mathtools,breqn,comment, graphicx}
\usepackage{xcolor,epstopdf}
\usepackage[caption=false]{subfig}

\theoremstyle{plain}
\newtheorem{theorem}{Theorem}[section]
\newtheorem{lemma}[theorem]{Lemma}

\theoremstyle{definition}

\theoremstyle{remark}

\newcommand\apj{\emph{The Astrophysical Journal}}%
\newcommand\apjl{\emph{The Astrophysical Journal Letters}}
\newcommand\aap{\emph{Astronomy and Astrophysics}}

\usepackage{amsmath,amsthm,mathrsfs}
\usepackage{graphicx,psfrag,epsf}
\usepackage{enumerate}
\usepackage{natbib}
\usepackage{url} 
\usepackage{amssymb,amsmath}
\usepackage{color,xcolor,nicefrac}

\def\cstat{${C}$ statistic}

\def\apj{Astrophys. J.}
\def\cmin{$C_{\mathrm{min}}$}

\def\cstat{${C}$ statistic}
\newcommand\Var{\text{Var}}

\newcommand\E{\text{E}}

\endlocaldefs

\begin{document}

\begin{frontmatter}
\title{Linear regression for Poisson count data: \\ A new semi--analytical method with applications to COVID-19 events}
\runtitle{Linear Poisson regression with application to COVID-19 data}

\begin{aug}
\author[A]{\fnms{Massimiliano} \snm{Bonamente}\ead[label=e1]{bonamem@uah.edu}},
\address[A]{Department of Physics and Astronomy,
University of Alabama in Huntsville, Huntsville AL 35899 (U.S.A), \printead{e1}}

\end{aug}

\begin{abstract}
This paper presents the application
  of a new semi---analytical method of linear regression for Poisson count data
  to  COVID-19 events. The regression is based on the Bonamente and Spence (2022) maximum--likelihood
   solution for the best--fit parameters, and this paper introduces a simple analytical solution for the
covariance matrix that completes the problem of linear regression with Poisson data.

  The analytical nature for both parameter estimates and their covariance matrix  is made possible
  by a convenient factorization of the linear model proposed by J. Scargle (2013).
  The method makes use of the asymptotic properties of the
  Fisher information matrix, whose inverse provides the covariance matrix.
  The combination of
  simple analytical methods to obtain both the maximum--likelihood estimates
  of the parameters, and their covariance matrix, constitute a new
  and convenient method for the linear regression of Poisson--distributed count data,
  which are of common occurrence across a variety of fields.

  A comparison between this new linear regression method and two
  alternative methods often used for the regression of count data --- the
  ordinary least--square regression and the $\chi^2$ regression -- is provided
  with the application of these methods to the analysis of recent COVID--19 count data.
The paper also discusses the relative advantages and disadvantages among these
 methods for the linear regression of Poisson count data.
\end{abstract}


\end{frontmatter}

\section{Introduction}
Integer--count data based on the Poisson distribution are quite common in real--life applications,
from S.-D. Poisson's use of his name--sake distribution for the study of civil and criminal cases
\citep{poisson1837}, Lord Rutherford's  characterization of
the rate of decay of radioactive isotopes \citep{rutherford1920} and
R.D. Clarke's application to the number of
bombs that fell in London during World~War~II \citep{clarke1946}, to
the  distribution of
low--count political events  \citep{king1988}, such as the number of seats lost by the president's party in a
mid--term election \citep{campbell1985}, or the
study of  mutations in bacteria \citep{luria1943}.
The common occurrence of the Poisson distribution lies primarily with its association
to the Poisson process, whereby this distribution describes the probability of
occurrence of events --- e.g., in a period of time or for a spatial region ---
from a process with a fixed rate of occurrence \citep[e.g.,][]{ross2003}.
Among the data analysis methods for Poisson--based count data,
the linear regression plays a fundamental role, given that the
linear model is arguably the simplest model that enforces a correlation
between the indepedendent variables (or regressors) and the dependent measured variables.

The maximum--likelihood estimate of a two--parameter
linear model for integer--valued count data
can be accomplished with the use of the Poisson log--likelihood, also known
as the \emph{Cash} or \cstat\  \citep{cash1979}.
Despite its common occurrence, it has been elusive to find a simple analytical expression for the
maximum--likelihood estimates of the parameters for the Poisson regression
(see, e.g., Sect.~2.2 of \citealt{cameron2013}). This restriction has resulted in a
limited use of the maximum--likelihood method for the regression of count data.
 This is in contrast to the simple
analytical solutions that are available for the ordinary least--squares (OLS)
regression
for data with equal variances (i.e, {homoskedastic} data) or with
different variances (i.e, {heteroskedastic} data), and for the
regression with the $\chi^2$ statistic,
which is a maximum--likelihood
method for heteroskedastic data of common use for normally--distributed data \citep[e.g.][]{greenwood1996}.
Such simple analytical solutions, which also extend to the more general multiple linear regression
with more than two adjustable parameters (e.g., see \citealt{bonamente2022book} for a review),
are now a standard of practice for the linear regression. Their use is in fact so widespread that
it is not uncommon to see $\chi^2$--based methods of regression even for
Poisson count data
for which the method is not justified, or rather it is known to lead to known biases
that have been documented as early as
\cite{lewis1949}, and more recently by \cite{kelly2007} and \cite{bonamente2020}.

Motivated by the need to provide a simple maximum--likelihood method
for the linear regression of count data (sometimes also referred to as \emph{cross--section} data)
that explicitly accounts for the Poisson distribution of
the measurements,
in \cite{bonamente2022} we have reported the first semi--analytical solution for the
parameters of the linear regression with the Poisson--based \cstat.
In this paper we continue the investigation of the same problem,
and report a new analytical solution for the
covariance matrix of the two model parameters, based on the Fisher information matrix
\citep{fisher1922b}.
The combination of the new results presented in this paper and those provided
in \cite{bonamente2022} thus provide a complete treatment of the maximum--likelihood
linear regression with Poisson--distributed count data.
The availability of these new analytical methods for  the linear regression of
count data therefore represents a key step towards a
wider use of an unbiased method of regression for such
integer--count data, without the need to resort to either numerical solutions
or the use of alternative and less accurate methods, such as those based on least--squares
or the $\chi^2$ statistic.

There are alternatives to the modelling of count data with the simple one--parameter Poisson
distribution considered in this paper. One of the key limitations of the Poisson regression is the inability
to properly model \emph{overdispersed} data (see, e.g., discussion in Chapter~3 of \cite{cameron2013}).
For such overdispersed data, other data models such as the negative binomial distribution can be used
\citep{hilbe2014}. Alternatively, the standard Poisson maximum--likelihood regression
can be modified to a \emph{quasi}--MLE regression with empirical variance functions that reflect
the overdispersion \citep[e.g.][]{cameron1986}, also used in the
generalized linear model (GLM) literature \citep[e.g.][]{mccullagh1989}. Nonetheless,
the simple linear regression of Poisson count data is an essential tool for the data analyst, and the
development of simple analytical tools for its application is the subject of this paper.

The paper is structured as follows:
Sect.~\ref{sec:model} summarizes the model for the regression with count data, which
is described in full in \cite{bonamente2022};
Sect.~\ref{sec:cov} describes methods to obtain the covariance matrix,
with Sect.~\ref{sec:properties} describing the properties of the matrix for the Poisson
linear regression under consideration and
Sect.~\ref{sec:errorProp} presenting an alternative method
to estimate the covariance matrix based on the error propagation method.
Sect.~\ref{sec:application} provides an application of the results
to recent COVID--19 data with comparison to the least--squares and
$\chi^2$ methods. Finally,  Sect.~\ref{sec:conclusions} contains a
discussion and the conclusions.

\section{Model for the regression of count data with the \cstat}
\label{sec:model}
The data under consideration are in the form of independent pairs of measurements
$(x_i, y_i)$, where $y_i$ are Poisson--distributed variables, for $i=1,\dots,N$, and
$x_i$ are fixed values of the independent variable.
\subsection{The \cite{scargle2013} parameterization of the linear model}
A convenient parameterization
of the linear model was proposed by \cite{scargle2013}, and it is in the form
of
\begin{equation}
f_S(x)=\lambda (1+a(x-x_A)),
\label{eq:scargle}
\end{equation}
where $x_A$ is a constant that usually coincides with the beginning of the
range of the independent variable $x$, and $\lambda$ and $a$ are the two parameters.
This parameterization is thus equivalent to a linear model with overall slope $\lambda\, a$,
and intercept $\lambda(1-a\,x_A)$, yet this factorization has algebraic andvantages when
used in the Poisson log--likelihood. The parent mean of the Poisson distribution
of the variable $y_i$
is then
\begin{equation}
  \mu_i= f_S(x_i) \, \Delta x_i,
  \label{eq:mui}
\end{equation}
  with $\Delta x_i$ the width of the $i$--th bin of data,
which is not required to be uniform among the measurements.

\subsection{Maximum--likelihood regression}
\label{sec:MLregression}
Upon minimization of the logarithm of the Poisson likelihood of the data with model,
in \cite{bonamente2022} we have shown that
the method of maximum likelihood yields two simple equations for the determination
of the parameter estimates. The first is a analytical relationship between
the two parameters,

\begin{equation}
  \lambda(a) = \frac{M}{R\left(1+a \dfrac{R}{2}\right)}
\label{eq:lambda}
\end{equation}
with $M$ the sum of all integer counts $y_i$, and $R$ the range of the independent
variable, typically chosed to be $R=x_B-x_A$.
This equations enables the determination of $\hat{\lambda}$ from $\hat{a}$. The second is
an equation that must be solved to obtain $\hat{a}$, and it is conveniently cast
as the zero of a function $F(a)$,
\begin{equation}
  F(a) = 1+\dfrac{R}{2} \left(a -\dfrac{M}{g(a)} \right) = 0,\, \text{ with }
  g(a) = \sum_{i=1}^N y_i \dfrac{(x_i-x_A)}{1+a(x_i-x_A)}.
\label{eq:Fa}
\end{equation}
In general, solution of \eqref{eq:Fa} to obtain the estimate $\hat{a}$ can be obtained numerically.
The inherently analytical method leading to \eqref{eq:lambda} and \eqref{eq:Fa},
and the need for a numerical solution of the latter, constitute this new \emph{semi--analytical}
method for the linear regression of Poisson count data.
The properties of the two functions $F(a)$ and $g(a)$ are described in detail in \cite{bonamente2022},
and the key properties are reported in the following.

\subsection{Identification of points of singularity and zeros of $g(a)$}

It is immediate to identify the $n$ points of singularity of $g(a)$ and its zeros,
 using
the property that $g'(a)<0$ between points of singularity.
Since it is also true that $F'(a)<0$ between its points of singularity, which are the zeros of $g(a)$,
the $n-1$ zeros of $F(a)$ can be easily found because of
the continuity of the function between singularities. In particular, the \emph{acceptable}
solution --- defined as the pair of best--fit parameters $(\hat{a}, \hat{\lambda})$
that gives a non--negative
model throughout its support, as required by the counting nature of the Poisson distribution ---
is either the first or the last zero, according to the
asymptotic
value of the function $F(a)$  for $a \to \pm \infty$.
The acceptable solution $\hat{a}$ can therefore be found with elementary numerical methods
from the equation $F(a)=0$.

\subsection{Applicability of model to data with gaps}
Both equations \eqref{eq:lambda} and \eqref{eq:Fa}
apply to data with arbitrary bin size. Moreover, when there are gaps in the data, defined as intervals of the
$x$ variable without measurements, the two equations are  modified with the introduction of
a modified range $R_m$ defined as
\begin{equation}
  R_m=\dfrac{R^2 - 2 S_G}{R-R_g}
  \label{eq:Rm}
\end{equation}
where $R_G$ is the sum of the ranges $R_{G,j}$ of all the gaps, and $S_G$ is a constant defined
by
\[ S_G = \sum_{j=1}^g R_{G,j}(x_{G,j}-x_A) \]
with $x_{G,j}$ the midpoint of the $j$--th gap.
The equations are modified to~\footnote{This modification is justified in  Lemma~6.7
of \cite{bonamente2022}.}
\begin{equation} \begin{cases}
  F(a) = 1 + \dfrac{a R_m}{2} - \dfrac{M R_m}{2 g(a)} = 0\\[10pt]
  \lambda (a) = \dfrac{M}{R\left(1+a \dfrac{R}{2}\right) -(R_G+a S_G)}
\end{cases}
\label{eq:gaps}
\end{equation}
Such modifications for the presence of gaps in the data do not
introduce any additional statistical complication to the analysis of either the
maximum--likelihood estimates
or the covariance matrix. Gaps in the data are quite common in data analysis,
for example when a range of wavelengths in the spectrum of an astronomical
source is unavailable \citep[e.g., because
of detector inefficiencies, as in ][]{ahoranta2020},
or when a time interval for the light--curve of a source is unobserved \citep[e.g.][]{levine1996}.

The reason for the parameterization of the linear model according to \eqref{eq:scargle} is that
the two resulting equations \eqref{eq:lambda} and \eqref{eq:Fa} can be used to identify
the maximum--likelihood parameter estimates more easily than for the standard parameterization
$a+b\,x$,
which leads to two coupled non--linear equations instead
(as noted, e.g., in Ch.~16 of \citealt{bonamente2022book}). Moreover, as will be shown
in the next section,
these equations lead to a very convenient analytical solution for the covariance matrix,
which represents the main new result of this paper.

\section{Evaluation of the covariance matrix}
\label{sec:cov}

\subsection{General properties of the covariance matrix}
The maximum--likelihood estimators
of the linear model parameters $\hat{\theta}=(\hat{\lambda},\hat{a})$
from Poisson count data
are  normally distributed, but only in the asymptotic limit of a large number of measurements $N$
(see, e.g., \citealt{eadie1971, cameron2013}).
In this asymptotic limit, the maximum--likelihood estimators are
also unbiased and efficient, with their variance reaching the \emph{Cram\'{e}r--Rao}
lower bound (see, e.g., \citealt{rao1945,cramer1946,amemiya1985}).
Within this limit, it is thus possible to approximate the covariance or error matrix of the parameters,
defined by
\[ \varepsilon  = \E [ (\hat{\theta}-\theta_0) (\hat{\theta}-\theta_0)^T],\]
as the inverse of the \emph{Fisher information matrix},
\begin{equation}
  \hat{\varepsilon} = - \left( \E\left[ \dfrac{\partial^2 \ln \mathscr{L}}{\partial \theta\, \partial \theta'}
  \right] \right)^{-1}_{\theta=\theta_0} = -\left( \E\left[ \sum_{i=1}^N \dfrac{\partial^2 \ln f(x_i)}{\partial \theta\, \partial \theta'}
  \right] \right)^{-1}_{\theta=\theta_0}.
  \label{eq:FisherInfo}
\end{equation}
In this application, $f(x_i)$ is the Poisson distribution for the measurement
at a value $x_i$ of the indepedendent variable, with a $\mu_i$ parameter
given by \eqref{eq:mui}, and the expectation in \eqref{eq:FisherInfo}
is taken with respect to  this distribution.~\footnote{The notation
$\partial^2 \ln \mathscr{L}/{\partial \theta\, \partial \theta'}$
indicates that diagonal elements are second derivatives with respect to the
same parameter, and off--diagonal
elements have cross--derivatives with respect to two different parameters. The symbol $^T$ indicates
the transpose.}
In practice, \eqref{eq:FisherInfo} means that the
maximum--likelihood estimates of the parameters of the
linear model converge in distribution to a normal distribution with the covariance matrix
given by \eqref{eq:FisherInfo}, or
\[
  (\hat{\theta}-\theta_0) \overset{d}\to N(0,\hat{\varepsilon}).
\]
R.~A. Fisher was the first to associate the concept of \emph{information} with  the
second derivative --- or curvature ---  of the  logarithm of the likelihood \cite{fisher1922b},
and similar definitions of information have been derived from this idea (e.g.,
\cite{shannon1949,kullback1951,akaike1974}).
Given that the true parameter values $\theta_0$ are unknown, the information
matrix is evaluated at $\theta=\hat{\theta}$, in recognition of the fact that the
maximum--likelihood estimates are unbiased and thus asymptotically converge to
the true--yet--unknown values (e.g., \citealt{fisher1922b}, Chapter~IX of \citealt{fisher1925},
or Chapter~33 of \citealt{cramer1946}).

\subsection{Evaluation of the covariance matrix for the Poisson linear regression with the Scargle et al. (2013) parameterization}
The negative of the logarithm of the likelihood, which is related to the \cstat\
by $C = -2 \ln \mathscr{L} + \text{const}$,  can be written as
\begin{equation}
 - \ln \mathscr{L} = \lambda (R-R_G) + \lambda a S_1 - M \ln \lambda - \sum_{i=1}^N y_i \ln (1+a(x_i-x_A)),
\end{equation}
where the constant $S_1$ is defined by
\begin{equation}
  S_1=\sum_{i=1}^N (x_i-x_A) \Delta x_i = \dfrac{R^2 - 2 S_G}{2},
  \label{eq:S1}
\end{equation}
and the second  equation connects $S_1$ with the constants defined in Sect.~\ref{sec:model},
and it is an immediate consequence of the model linearity.
The first derivatives are
\begin{equation}
  \begin{cases}
  -  \dfrac{\partial \ln \mathscr{L}}{\partial a} = \lambda S_1 - g(a)\\[10pt]
  -  \dfrac{\partial \ln \mathscr{L}}{\partial \lambda} = (R-R_G) +a S_1 - \dfrac{M}{\lambda}
  \end{cases}
\end{equation}
which, when set to zero, yield the usual maximum--likelihood estimates
described in Sect.~\ref{sec:model} and in more detail in \cite{bonamente2022}. From these, it is immediate to
evaluate the second order derivatives and thus obtain
\begin{equation}
  - \dfrac{\partial^2 \ln \mathscr{L}}{\partial \theta\, \partial \theta'} = \begin{bmatrix}
    \dfrac{M}{\lambda^2} & S_1 \\[10pt]
    S_1   & G(a)
  \end{bmatrix}
  \label{eq:matrix1}
\end{equation}
with a new function $G$ defined as
\begin{equation}
  G(a) = - \dfrac{ \partial g(a)}{\partial a} = \sum_{i=1}^N \dfrac{y_i (x_i-x_A)^2}{(1+a(x_i-x_A))^2}.
  \label{eq:Ga}
\end{equation}
The expectation of \eqref{eq:matrix1} is immediately calculated with the consideration that
\[ \left.\E[y_i]\right|_{\hat{\theta}} = \left.\mu_i\right|_{\hat{\theta}} = \hat{\lambda} (1+\hat{a}(x_i-x_A)) \Delta x_i,
\]
where $\Delta x_i$ is the size of the $i$--th bin of data.
Given that all terms in \eqref{eq:matrix1} are linear in the measurements $y_i$, the expectation is found
as
\begin{equation}
   - \left. \E\left[\dfrac{\partial^2 \ln \mathscr{L}}{\partial \theta\, \partial \theta'} \right]
  \right|_{\hat{\theta}} =
  \begin{bmatrix}  \dfrac{(R-R_G) +\hat{a} S_1}{\hat{\lambda}} & S_1 \\[10pt]
    S_1 & \hat{\lambda}\, H(\hat{a})
  \end{bmatrix},
  \label{eq:matrix2}
\end{equation}
with
\begin{equation}
  H(a) = \sum_{i=1}^N \dfrac{(x_i-x_A)^2}{1 + {a}(x_i-x_A)} \Delta x_i
  \label{eq:Ha}
\end{equation}
defined as a convenient analytical function to describe one of the terms of the information
matrix. Notice how the bin size appears explicitly in this function, unlike
in the case of the function $g(a)$ in \eqref{eq:Fa}, which is independent of the bin size.

With the simple analytical expression available for the second--order derivatives
of the logarithm of the likelihood \eqref{eq:matrix2}, it is finally possible
to provide the asymptotic estimate of the covariance matrix as its inverse,
\begin{equation}
  \hat{\varepsilon} =
  \begin{bmatrix} \hat{\sigma}^2_a & \hat{\sigma}^2_{a\,\lambda} \\[5pt]
    \hat{\sigma}^2_{a\,\lambda} & \hat{\sigma}^2_{\lambda}
\end{bmatrix} =
  \dfrac{1}{\Delta} \begin{bmatrix}
    \hat{\lambda}\, H(\hat{a}) & -S_1 \\[10pt]
    -S_1 & \dfrac{(R-R_G) + \hat{a} S_1}{\hat{\lambda}}
  \end{bmatrix},
  \label{eq:matrix3}
\end{equation}
with $\Delta$ the determinant of the matrix \eqref{eq:matrix2},
\[ \Delta = \left((R-R_G) + \hat{a} S_1\right) H(\hat{a}) - S_1^2.
\]
This simple analytical form for the covariance matrix of the maximum--likelihood
estimates for the linear regression to Poisson data is the main result of this
paper. Its properties are analyzed in the following section.

\subsection{Properties of the  covariance matrix}
\label{sec:properties}
Basic results on the sign of the terms in the covariance matrix \eqref{eq:matrix3}
are provided in this section. First, it is necessary to recall the definition of
\emph{acceptability} of the estimates of the model parameters as values that
ensure a non--negative parent Poisson mean, which was briefly introduced in Sect.~\ref{sec:model}.
Acceptable solutions for $\hat{a}$ and
$\hat{\lambda}$ require that the quantities
\[ \hat{\lambda}(1+\hat{a}(x_i-x_A))\]
are all non--negative. This results in values of the two parameters to be found in two
possible ranges, namely (a) for $\hat{a} \leq -2/\Delta x_1$, where $\hat{\lambda}<0$, or
(b) for $\hat{a}>-1/(R-\Delta x_N/2)$, where $\hat{\lambda}>0$ (see Lemma~5.2
of \citealt{bonamente2022} for a proof). In the following, it is assumed that $x_A$ is
the beginning of the range of the independent variable $x$.

\begin{lemma}[Sign of the determinant $\Delta$]
  The determinant $\Delta$ in \eqref{eq:matrix3} is  positive definite for all acceptable solutions
  of the regression coefficients.
\end{lemma}
\begin{proof}
The determinant can be written as
  \[ \Delta = \sum_{i=1}^N \dfrac{(x_i-x_A)^2}{1+\hat{a}(x_i-x_A)} \Delta x_i \cdot
  \sum_{i=1}^N  (1+\hat{a}(x_i-x_A))  \Delta x_i -
  \left( \sum_{i=1}^N  (x_i-x_A) \Delta x_i \right)^2.
  \]
This is in the form
  \[ \Delta = \sum_{i=1}^N \dfrac{r_i^2}{s_i} \cdot\sum_{i=1}^N s_i - \left( \sum_{i=1}^N r_i \right)^2
  \]
  with $r_i=(x_i-x_A) \Delta x_i$ and $s_i=(1+\hat{a}(x_i-x_A)) \Delta x_i$. Notice that
the terms $s_i$ are either all positive or all negative, therefore making it such that
one can use $|s_i|$ in their place in the previous equation,
  \[ \Delta = \sum_{i=1}^N \dfrac{r_i^2}{|s_i|} \cdot\sum_{i=1}^N |s_i| - \left( \sum_{i=1}^N r_i \right)^2.
  \]
  With the substitution $a_i=r_i/\sqrt{|s_i|}$ and $b_i=\sqrt{|s_i|}$, the determinant becomes
  \[
   \Delta = \sum_{i=1}^N a_i^2 \cdot \sum_{i=1}^N b_i^2 - \left( \sum_{i=1}^N a_i b_i\right)^2 \geq 0,
 \]
 according to the Cauchy--Schwartz inequality.
\end{proof}

It is also possible to show that the variances of $\hat{a}$ and $\hat{\lambda}$ have the proper sign,
and that the covariance between them is negative.
\begin{lemma}[Signs of variance and covariance terms] The diagonal terms in the covariance matrix
  \eqref{eq:matrix3} are positive, and the off--diagonal terms are negative, for all acceptable solutions
  of the linear regression coefficients.
\end{lemma}
\begin{proof}
  The term $S_1$ is positive according to its definition \eqref{eq:S1}, and therefore
  the covariance is always negative.

According to \eqref{eq:Ha}, the variance of the estimate $\hat{a}$ is
  \[ \hat{\sigma}^2_a = \dfrac{\hat{\lambda} H(\hat{a})}\Delta
  =\dfrac{\hat{\lambda}}{\Delta}  \sum_{i=1}^N \dfrac{(x_i-x_A)^2}{1 + {a}(x_i-x_A)} \Delta x_i. \]
  Given the discussion at the beginning of this section, $\hat{\lambda}$ and each of the terms
  $1+\hat{a}(x_i-x_A)$ must have the same sign for acceptable values of the regression
  coefficients. The variance is therefore always positive, as it should.

  For the same reasons, the variance of $\hat{\lambda}$ is immediately shown to have the proper sign,
  \[ \hat{\sigma}^2_{\lambda} = \dfrac{1}{\Delta \hat{\lambda}} \left( (R-R_G) + \hat{a} S_1\right)
  = \dfrac{1}{\Delta} \dfrac{\sum_{i=1}^N (1+\hat{a}(x_i-x_A)) \Delta x_i}{\hat{\lambda}} \geq 0.
  \]

  Finally, the negative sign of the covariance term is an immediate consequence of the definition of $S_1$
  from \eqref{eq:S1}.
\end{proof}

Al alternative approach to achieve a non--negative Poisson mean for all values of the model parameters
would be to perform an exponential transform,
so that the expectations for the Poisson--distributed mesurements is, for example, of the type
$\mu_i = \exp \left({a+b\,x_i}\right)$ (see, e.g. Sect.~2.2 of \citealt{cameron2013}).
In practice, this corresponds to a linear model for the logarithm of the means,
$\log \mu_i = a + b\, x_i$.
Although this transformation does indeed provide a Poisson mean $\mu_i$ that is always non--negative,
the model is no longer \emph{linear} between $x$ and  $y$, and therefore unsuitable for those cases where
a linear relationship is required, e.g., because of the physical nature of the two variables.
Moreover, even with this exponential transformation, the maximum--likelihood method of regression
for Poisson data does not have an analytical solution with the usual parameterization of the linear
model \cite{cameron2013}.

\subsection{An alternative approximation for the covariance matrix}
\label{sec:errorProp}

A different analytical approximation for the covariance matrix of the
 coefficients of the linear regression can be given using the \emph{error propagation}
 method, also known as the \emph{delta} method \citep[e.g.][]{bevington2003, cameron2013}.
 The method consists of treating the parameters as a function of the
 independent measurements, $\theta_j=\theta_j(y_k)$,
 and approximating the variances and covariances with the first--order terms of their
 Taylor series,
 \begin{equation}
   \hat{\sigma}^2_{\theta_i\,\theta_j} \simeq \sum_{k=1}^N \left. \left(\dfrac{\partial \theta_i}{\partial y_k}\right)
   \right|_{\hat{\theta}}
   \left.\left(\dfrac{\partial \theta_j}{\partial y_k}\right) \right|_{\hat{\theta}} \sigma^2_k,
   \label{eq:errProp}
 \end{equation}
where $\sigma^2_k$ is an estimate of the variance of the measurement $y_k$.
The regression model described in Sect.~\ref{sec:model} does not lead to best--fit parameters $\hat{\theta}$
that are described as a function of the measurements in a closed form. Nonetheless,
it is possible to take the derivatives of both \eqref{eq:lambda} and \eqref{eq:Fa}, or more generally
\eqref{eq:gaps},
with respect to $y_j$ to obtain the derivatives $\partial \theta_i / \partial y_k$ as follows.

Start with
\[ \lambda R \left(1+a \dfrac{R}{2}\right) - \lambda(R_G+aS_G) = M
\]
which is the most general equation to obtain $\hat{a}$ in the presence of
gaps in the data and of non--uniform binning. The equation can be brought
in a simpler form as
\[ \lambda(R-R_G) + a \lambda\left( \dfrac{R^2}{2} - S_G \right) = \sum_{i=1}^N y_i
\]
leading to
\[ \dfrac{\partial \lambda}{\partial y_j} \left( R -R_G + a\left( \dfrac{R^2}{2} - S_G \right) \right)
+ \dfrac{\partial a}{\partial y_j} \lambda \left(\dfrac{R^2}{2} - S_G\right) = 1\]
and therefore
\begin{equation}
  \dfrac{\partial \lambda}{\partial y_j} = \dfrac{ 1 - \lambda\left(\dfrac{R^2}{2} - S_G\right)  \dfrac{\partial a}{\partial y_j}}
  { R -R_G + a\left( \dfrac{R^2}{2} - S_G \right)}.
\end{equation}
Alternatively, this partial derivative can be expressed in terms of $a$ alone as
\begin{equation}
  \dfrac{\partial \lambda}{\partial y_j} = \dfrac{R-R_G + \left(\dfrac{R^2}{2} - S_G\right)\left(a-M \dfrac{\partial a}{\partial y_j}\right)}{\left( R -R_G + a\left( \dfrac{R^2}{2} - S_G \right)\right)^2},
  \label{eq:dlambdady}
\end{equation}
where the correction for the presence of gaps in the data is in the constants $R_G$ and $S_G$, both of which become
null when the data spans a continous range between $x_A$ and $x_B$, with $R=x_B-x_A$.

A similar procedure yields the partial derivatives of $a$ from
\[ \left(a+\dfrac{2}{R_m}\right)g(a) = M, \]
which is the general expression corresponding to the equation $F(a)=0$ needed to find $\hat{a}$.
The equation can be written in a more convenient way as
\[ \sum_{i=1}^N \dfrac{ y_i \left(  2/R_m (x_i-x_A) -1\right)}{1 + a(x_i-x_A)} = 0\]
and from this, taking the derivative with respect to $y_j$,
\begin{equation}
  \dfrac{\partial a}{\partial y_j} = \dfrac{1- 2/R_m (x_j-x_A)}{1+a(x_j-x_A)} \cdot  \dfrac{1}{g_2(a)-2/R_m G(a)},
  \label{eq:dady}
\end{equation}
with the aid of a new function defined by
\begin{equation}
 g_2(a)=\sum_{i=1}^N \dfrac{y_i (x_i-x_A)}{(1+a(x_i-x_A))^2}.
  \label{eq:g2}
\end{equation}
The result is proven by noting that, in taking the total differential of the term in the sum for $i=j$,
there is an extra term due to the presence of the random variable $y_j$.

The partial derivatives in \eqref{eq:dlambdady} and \eqref{eq:dady}
can now be used in \eqref{eq:errProp} to provide an estimate of the
covariance matrix from the error propagation method. The variance  $\sigma^2_k$
of the measurement $y_k$ can be estimated as either the measurement $y_k$ itself, or
as the best--fit model $\hat{\mu}_k=\E[y_k/\hat{\theta}] = \hat{\lambda}(1+\hat{a}(x_k-x_A)) \Delta x_k$.
In the application provided below in Sect.~\ref{sec:application}, the parent
expectation $\hat{\mu}_k$ is used for the numerical result.

\section{Application to COVID--19 data}
\label{sec:application}
The daily number of deaths caused by the COVID--19 pandemic in the United States
provides an example
of count data that can be analyzed with the linear regression model presented in this paper.
The data are obtained from the \emph{New York Times} \citep{NYTCovid}, and
they are reported in Table~\ref{tab:covid}.~\footnote{The \emph{New York Times} archive
reports the cumulative number of deaths, Table~\ref{tab:covid} also reports the daily counts.}
These data follow the data model of Sect.~\ref{sec:model}, whereby the number of events
are collected in a time interval with a uniform size of length $\Delta x =~1$~day,
and the measurements in the column `Number of events' are independent of one another.
\begin{table}
  \centering
  \caption{Daily number of events (deaths) reported at the beginning
  of the pandemic. Data are from \cite{NYTCovid}. Day 0 is February 28, 2020,
  the day prior to the first reported death, which occurred on February 29, 2020.}
  \begin{tabular}{lcc}
  \hline
  \hline
    Day & Number of events & Cumulative\\
  \hline
    0 & 0  & 0\\
    1 & 1  & 1 \\
    2 & 2  & 3 \\
    3 & 3  & 6 \\
    4 & 4  & 10 \\
    5 & 2  & 12 \\
    6 & 0  & 12 \\
    7 & 3 & 15 \\
    8 & 4 & 19 \\
    9 & 3 & 22 \\
    10 & 4 & 26 \\
    11 & 5 & 31 \\
    12 & 6 & 37 \\
    13 & 6 & 43 \\
    14 & 7 & 50 \\
    15 & 10 & 60 \\
    16 & 8 & 68 \\
    17 & 23 & 91 \\
    18 & 26 & 117 \\
    19 & 45 & 162 \\
    \multicolumn{3}{c}{\dots}\\
    \hline
  \hline
\end{tabular}
\label{tab:covid}
\end{table}
This section presents two sets of linear regression to these data with the \cstat,
to illustrate the novel method of analysis
described in this paper and in \cite{bonamente2022}. It also presents the comparison
with two popular regression methods,
the ordinary least squares regression and the fit based on the $\chi^2$ statistic.
The comparison among these methods is used to illustrate the advantages of the new method
of linear regression for Poisson data, and to discuss assumptions and limitations of the different methods.

\begin{figure}[t]
  \centering
  \includegraphics[width=5in]{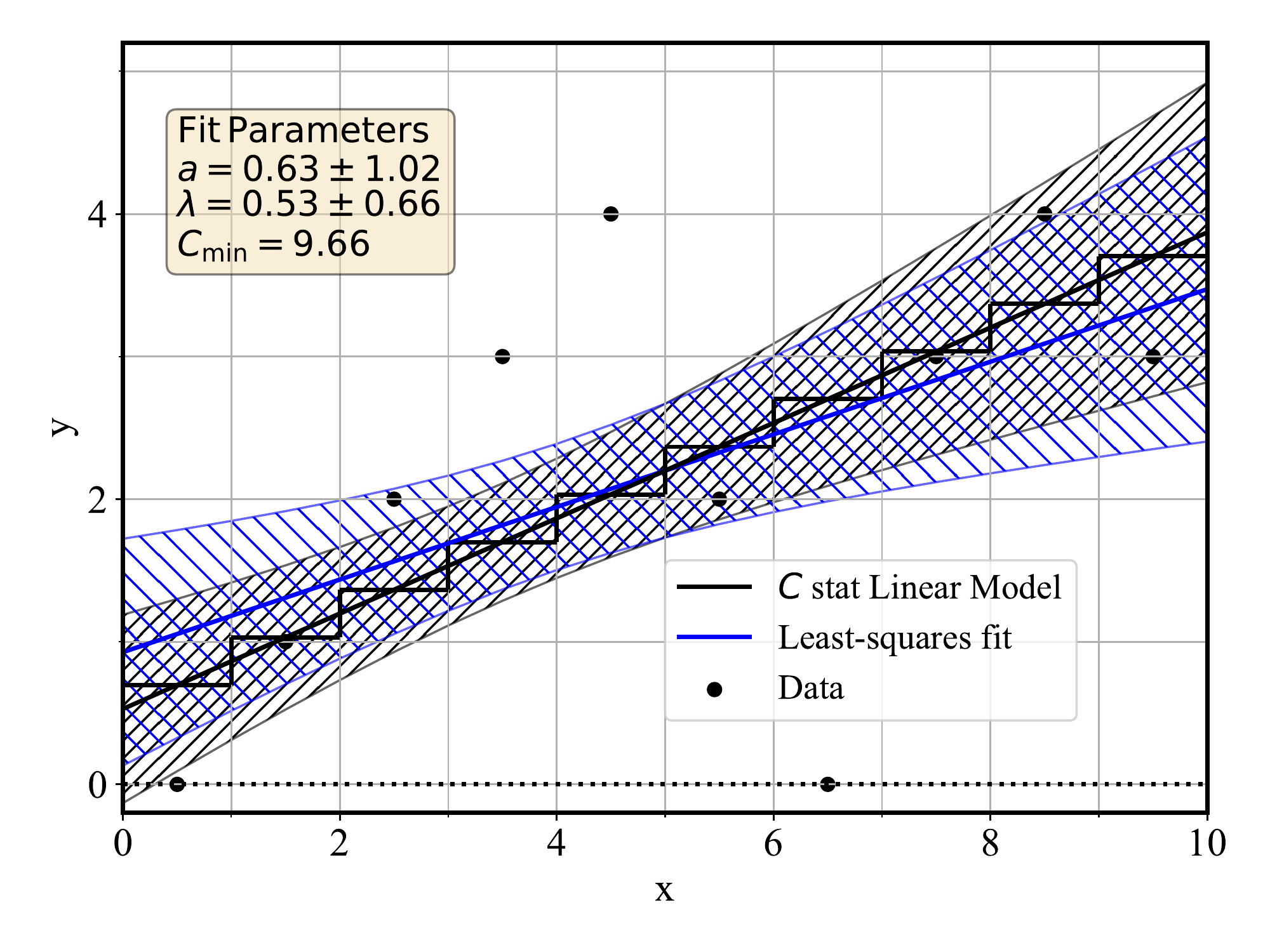}
  \caption{Linear regressions to days 0-9 of the COVID-19 data of Table~\ref{tab:covid}. The black curves correspond to the
  linear regression with the \cstat, and the blue curves for the ordinary least--squares regression.}
  \label{fig:fit1}
\end{figure}

\subsection{Fit to days 0--9 with uniform binning}
\label{sec:fitPoisson}
For this application, the data
for days 0 (day prior to first case recorded) through 9 from Table~\ref{tab:covid}
are fit to the linear model, with a total of $N=10$ data points.
\subsubsection{The maximum--likelihood parameter estimates}
The best--fit model according to \eqref{eq:scargle}
is shown in Figure~\ref{fig:fit1}
as the black line, with the step--wise continuous black line indicating the
best--fit model $\hat{\mu}_i$. The two lines intersect because the data are collected
over bins with unit size ($\Delta x=1$).

The method of maximum--likelihood regression (see Sect.~\ref{sec:model}) makes use of two functions
$g(a)$ and $F(a)$ that are reported in Figure~\ref{fig:gaFa}.
\begin{figure}
  \centering
  \includegraphics[width=5in]{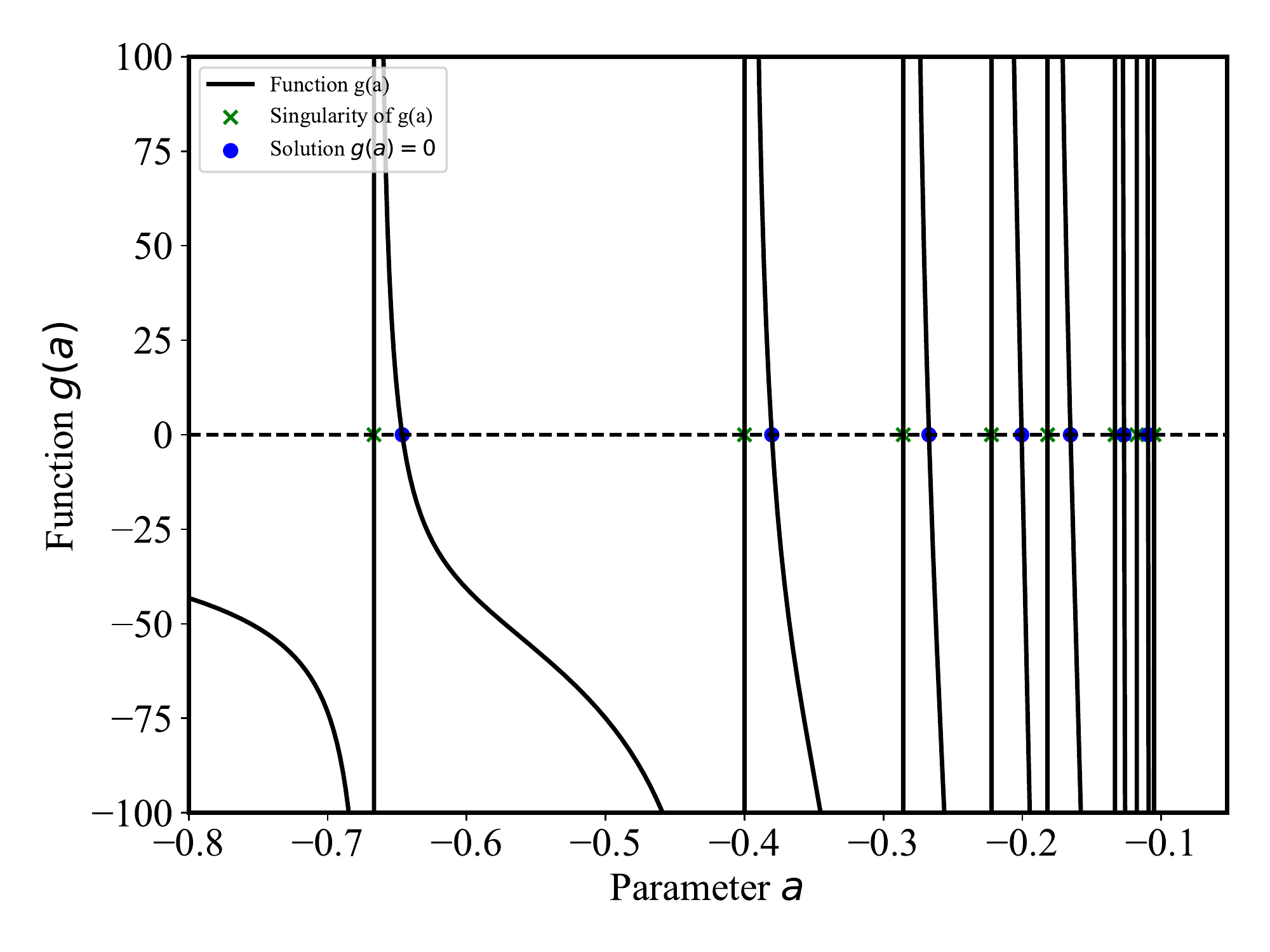}
  \includegraphics[width=5in]{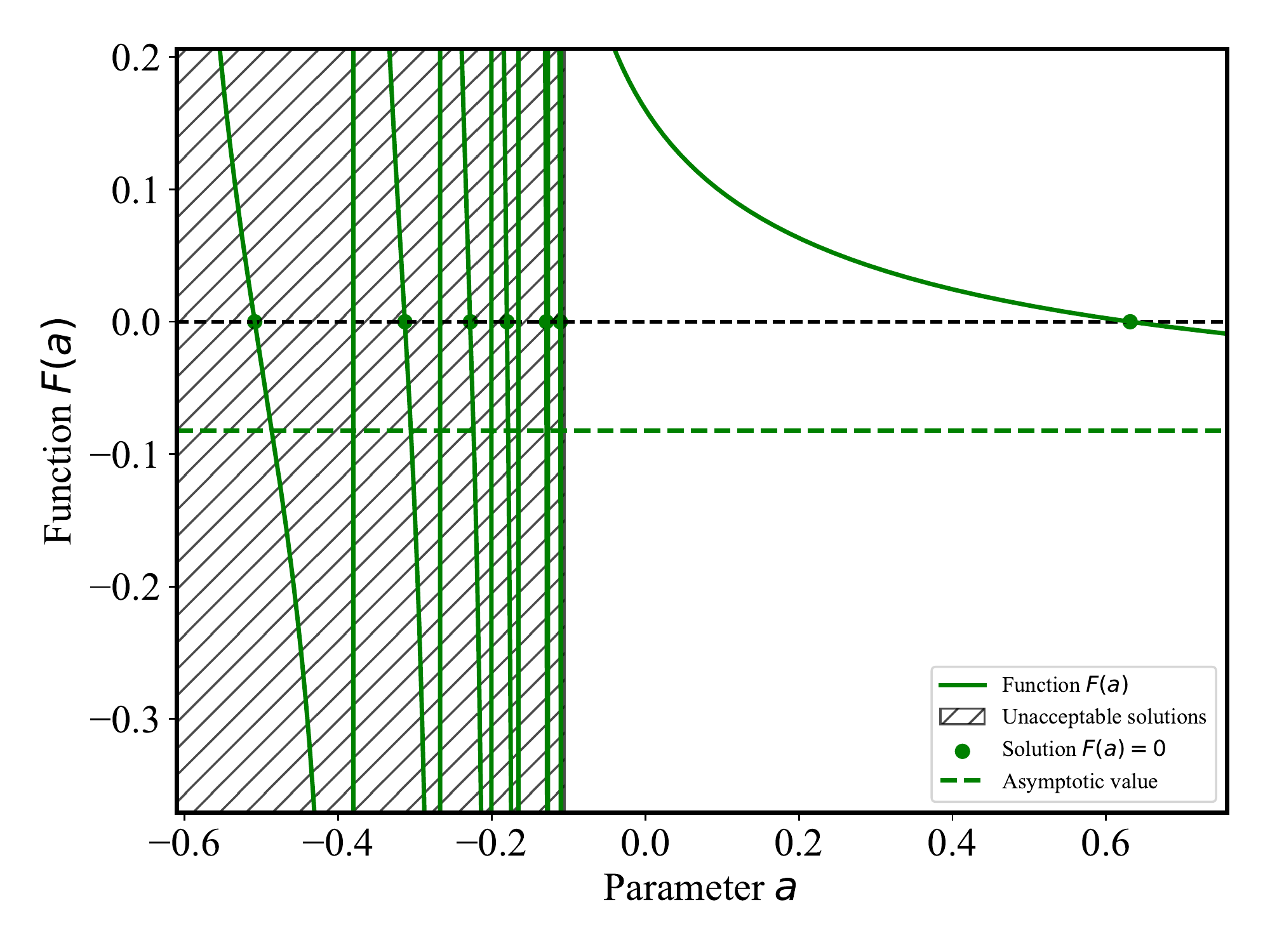}
  \caption{Functions $g(a)$ and $F(a)$ used for the Poisson regression, see \eqref{eq:Fa}.}
  \label{fig:gaFa}
\end{figure}
In this application, there is a total of $M=22$ counts for $N=10$ bins and $n=8$ unique non--zero bins.
According to \eqref{eq:Fa}, this results in $n=8$ points of singularity for $g(a)$ marked as
green crosses in Figure~\ref{fig:gaFa}, and $n-1=7$ roots or zeros of $g(a)$ marked as blue dots between consecutive
singularities, with $g'(a)<0$ in each interval. The roots of $g(a)$ are also points of singularity for $F(a)$, whose
zero(s) are the estimates of the parameter $a$. Of the $n-2=6$ such solutions, only the last one is \emph{acceptable},
i.e., it is the only value $\hat{a}$ that gives a non--negative model throughout its support, as required for
Poisson counts (see Sect.~5 of \citealt{bonamente2022}). Having identified the last singularity of $F(a)$ as the last zero of $g(a)$, the acceptable
solution $\hat{a}$ is marked as the rightmost green dot in Figure~\ref{fig:gaFa}.

It is clear that it is not necessary to identify and calculate all the roots of $g(a)$ in order to determine $\hat{a}$, but just the last one
between the last two singularities of $g(a)$. The computational burden is therefore very limited, even for
data with a large number of events. Figure~\ref{fig:gaFa} report all points of sngularities and the zeros for both $g(a)$ and $F(a)$
only for illustration purposes.

\subsubsection{The covariance matrix}
The error matrix evaluated according to
\eqref{eq:matrix3} is
\[
  \hat{\varepsilon}=\begin{bmatrix} 0.346 & -0.665 \\
     -0.665 & 1.044
  \end{bmatrix}
\]
with the estimated standard deviations of the parameters also reported
in the figure, for estimated parameters $\hat{a}=0.63\pm1.02$ and $\hat{\lambda}=0.53\pm0.66$.

\subsubsection{The log--likelihood or $C$ statistic}
\label{sec:cstatHypothesis}
Moreover, the fit statistic is evaluated according to
\begin{equation}
  C_{\mathrm{min}} = \sum_{i=1}^N (\hat{\mu}_i - y_i + y_i \ln \left(y_i/\hat{\mu}_i\right)),
\label{eq:cmin}
\end{equation}
as reported in Equation~3~of \cite{bonamente2022}.
Hypothesis testing with the \cstat\ requires the evaluation of
critical values of its distribution, which is known only in the
asymptotic limits of a large number of measurements $N$ and for large
parent means, when \cmin\ is approximately distributed
as $\chi^2(N-m)$ \citep[see, e.g.,][]{lampton1976, bonamente2022book},
with $m=2$ adjustable parameters.
Although a detailed discussion of the method of hypothesis testing with
the \cstat\ goes beyond the scope of this paper, it is nonetheless
useful to provide a brief description of its application for this example,
since hypothesis testing is an integral part of the overall regression.
In the present example of small
Poisson means, the approximation of \cmin\ with a $\chi^2(N-m)$ distribution
is not accurate, and in  \cite{bonamente2020} we proposed an approximate
method to estimate critical values that applies also to the low--mean case.
The method consists of
using the parent mean and variance of each term in \eqref{eq:cmin}
in place of the mean and variance of a $\chi^2(1)$ distribution
(which are respectively 1 and 2),
and use the central limit theorem to ensure an approximately
normal distribution for the statistic. Following this method,
the critical value of the \cmin\ distribution for $N-m=8$ degrees of freedom
is
13.6 (compared to a corresponding critical value of 13.4 for
the $\chi^2(8)$ distribution), and therefore these data are
consistent with the null hypothesis of a linear model, at the 90\% confidence level.
Notice how both the parameter estimation and the goodness--of--fit statistic
naturally account for measurements with zero counts, such as bins 0 and 6.

\subsubsection{The confidence band}
The grey hatched band associated with the best--fit
model highlights another strength of the availability of the
error matrix for the \cstat\ regression, which makes it is possible to estimate the variance of any function of the model parameters.
In particular, it is often useful to estimate the variance
of the \emph{expectation} of the $f_S(x)$ function for a fixed value of the independent variable $x$, i.e.,
\[ \hat{f}_S(x) = \E[f_S(x)/ x, \hat{\theta}] = \hat{\lambda} (1+\hat{a}(1-x)),\]
or, equivalently, for the expectation of the number of counts in a bin centered at $x$ and with
arbitrary bin size,
\[ \hat{y}_x=\E[y/ x, \hat{\theta}] = \hat{\lambda} (1+\hat{a}(1-x))\,\Delta x,
\]
where $x$ is any value of the regressor variable, not limited to the values $x_i$.
The variance of the variable $\hat{y}_x$ is represented by the variance of $y$ values along a vertical line, for a given $x$. The range
 $\hat{y}_x \pm \sigma_{y_x}$ is usually referred to as the \emph{confidence band} of the best--fit model,
 and is represented by the grey hatched region in Fig.~\ref{fig:fit1}. The variance $\sigma^2_{y_x}$
is obtained
by the error--propagation method, with
\[ \hat{\sigma}^2_{y_x} = \hat{\sigma}^2_{\lambda} \left.\left(\dfrac{\partial y_x}{\partial \lambda} \right)^2\right|_{\hat{\theta}} +
\hat{\sigma}^2_{a} \left.\left(\dfrac{\partial y_x}{\partial a} \right)^2\right|_{\hat{\theta}}
+  2\, \hat{\sigma}^2_{\lambda\,a}  \left.\left(\dfrac{\partial y_x}{\partial \lambda}\right)  \left(\dfrac{\partial y_x}{\partial a}  \right)\right|_{\hat{\theta}}.\]
The standard error $\sigma_{y_x}$  is therefore estimated immediately
from the error matrix $\hat{\varepsilon}$ in \eqref{eq:matrix3} and from the
analytical form for $y_x$.
Similar condiderations can be extended to any function of the model parameters. Another
useful application is to
the overall slope of the model, which is measured as $\hat{a}\cdot \hat{\lambda}=0.33 \pm 0.15$ in this example.

\subsection{Comparison to OLS regression}
\label{sec:OLS}
For comparison, Fig.~\ref{fig:fit1} also reports
the ordinary least--squares fit to the same data, with uniform weights for all the points,
with the linear model following the usual parameterization $f(x) = a + b x$.
\subsubsection{OLS estimators}
The OLS estimators are given by the usual formulas,
\begin{equation}
  \hat{\theta}_{\mathrm OLS} = \begin{bmatrix} \hat{a} \\[10pt] \hat{b} \end{bmatrix} = \dfrac{1}{\Delta}
  \begin{bmatrix} \sum_{i=1}^N x_i^2 \cdot \sum_{i=1}^Ny_i -  \sum_{i=1}^N x_i \cdot  \sum_{i=1}^N x_i\,y_i \\[10pt]
N \sum_{i=1}^N x_i \; y_i - \sum_{i=1}^N x_i  \cdot \sum_{i=1}^Ny_i
  \end{bmatrix}
\end{equation}
with
\[
\Delta = N \sum_{i=1}^N x_i^2 - \left(\sum_{i=1}^N x_i\right)^2.
\]
The OLS regression is known to provide an
 unbiased estimate of the parameters, regardless of the parent distribution
of the measurements, due to the linearity of the model \citep[see, e.g.][]{eadie1971,cameron2013}. As a result,
Fig.~\ref{fig:fit1} shows that the \cstat\ and the OLS linear regressions provide qualitatively similar
best--fit values.

\subsubsection{Estimate of covariance matrix for OLS}
Under the assumption of a normal distribution with equal variance for all the measurements,
the OLS method is equivalent to the maximum--likelihood method, and
 the value for the common
variance is required to estimate the covariance matrix \citep[e.g., see Chapter 11 of][]{bonamente2022book}.
For measurements that are not normally distributed, such as the data model presented in this paper and for this specific
data example,
it is still possible to provide an estimate of the parameter variances for the OLS estimators, \emph{if}
the OLS estimator are interpreted as a function of the measurements $y_i$, without
requiring homoskedasticity  or normality. In this case, the linear OLS is no longer a maximum--likelihood
estimator, but it retains the convenient property of unbiasedness.

Using the assumption that $y_i \sim \text{Poisson}(\theta_0+x_i \theta_1)$ and that the
measurements are independent,
it follows that $\Var(y_i)=\theta_0+x_i \theta_1$, and therefore
\begin{equation}
  \Var(\hat{\theta}_{\mathrm OLS}) \simeq \dfrac{1}{\Delta^2}
  \begin{bmatrix}  \sum_{i=1}^N \left( \sum_{i=1}^N x_i^2 - x_i \sum_{i=1}^N x_i\right)^2 \cdot (\theta_0+x_i \theta_1) \\[10pt]
    	 \sum_{i=1}^N \left( N x_i - \sum_{i=1}^N x_i\right)^2 \cdot (\theta_0+x_i \theta_1)
     \end{bmatrix}
     \label{eq:OLSPoissVar}
\end{equation}
Naturally, since the true parameters are unknown, one can only use these variances by substituting $\hat{\theta}$ for $\theta$
in the right--hand side of \eqref{eq:OLSPoissVar}. This simple result is due to the linearity of the model, and the independence
of the measurements.

It is necessary to emphasize that the approximate variances derived
in \eqref{eq:OLSPoissVar} are an \emph{ad hoc} estimate, and they do
differ from the standard OLS error matrix~\footnote{See, e.g., Eq.~11.19 of
\cite{bonamente2022book} or Example 19.6 of \cite{kendall1973}.}
\begin{equation}
  \varepsilon_{\mathrm OLS} = \dfrac{\sigma^2}{\Delta}  \begin{bmatrix}
     \sum_{i=1}^N x_i^2 & - \sum_{i=1}^N x_i\\[10pt]
    \dots & N \\
  \end{bmatrix}
  \label{eq:OLSError}
\end{equation}
where $\sigma^2$ is the common (and unknown) variance. In the usual case of the OLS with equal variances,
one would normally proceed with the estimation of $\sigma^2$ via the sum of square residuals, divided by $N-k$,
\begin{equation}
  \hat{\sigma}^2 = \dfrac{1}{N-k} \sum_{i=1}^N (y_i - (\hat{a} + \hat{b}\, x_i))^2
  \label{eq:s2}
\end{equation}
where $k=2$ corresponds to the two adjustable parameters of the model. But such estimator would be biases
for the case of heteroskedastic data, and therefore not meaningful for this Poisson regression.
This is the reason for the development of the ad--hoc approximation \eqref{eq:OLSPoissVar}.

Moreover, an estimate for the covariance of the linear OLS estimators
can be obtained by the usual error propagation method, namely
\[ \hat{\sigma}^2_{ab} = \sum_{j=1}^N \dfrac{\partial \hat{a}}{\partial y_j} \dfrac{\partial \hat{b}}{\partial y_j} \sigma^2_j\]
where $\sigma^2_j$ is the variance of the $j$--th measurement, in this application assumed to be Poisson--distributed.
Notice that this method can also be used to obtain the error matrix \eqref{eq:OLSError}
(e.g., see Sect.~11.3 of \citealt{bonamente2022book}.)
The result is
\begin{equation}
    \hat{\sigma}^2_{ab}  = \small \dfrac{1}{\Delta^2} \left( \left(N \sum_{i=1}^N  x_i^2 + \left(\sum_{i=1}^N  x_i\right)^2 \right) \sum_{j=1}^N x_j \sigma^2_j
     - \sum_{i=1}^N x_i \left( \sum_{i=1}^N  x_i^2 \sum_{j=1}^N \sigma^2_j +
    N \sum_{j=1}^N x_j^2 \sigma^2_j \right) \right) \label{eq:OLSPoissCov}
\end{equation}
which becomes the same as in \eqref{eq:OLSError} under the assumption of homoskedasticity.
In \eqref{eq:OLSPoissCov}, however, the variances to be used are the non--uniform
Poisson--estimated $\sigma^2_j = \hat{a}+x_j \hat{b}$.
In summary, it is still possible to estimate a covariance matrix for the linear OLS,
in place of the maximum--likelihood and Poisson--based \eqref{eq:matrix3}. These estimates, \eqref{eq:OLSPoissVar} and \eqref{eq:OLSPoissCov},
are obtained by a \emph{post facto} use of the Poisson distribution in the usual OLS
regression estimators, which in general do not require the specification of a parent distribution.
The data of Fig.~\ref{fig:fit1} yield linear OLS parameters of $a=0.93\pm0.80$, $b=0.25\pm0.16$,
for a covariance $\hat{\sigma}^2_{ab}=-0.11$, using \eqref{eq:OLSPoissVar}
and \eqref{eq:OLSPoissCov}. The confidence band according to these errors
is  reported as the blue hatched area in Fig.~\ref{fig:fit1}.

For the purpose of comparison,
the parameter uncertainties for the standard OLS regression according to \eqref{eq:OLSError} and \eqref{eq:s2} are also calculated.
First, from the OLS best--fit parameters \eqref{eq:s2} yields an estimated common data variance of $\hat{\sigma}^2=1.78$, and accordingly the
OLS error matrix yields errors that result in $a=0.93\pm0.85$, $b=0.25\pm0.15$.
In this application, the errors calculated for this
`standard' OLS model is similar to those of the 'modified' OLS model ($a=0.93\pm0.80$, $b=0.25\pm0.16$), and also
to those with the `proper' Poisson regression of Sect.~\ref{sec:fitPoisson} ($\hat{a}=0.63\pm1.02$ and $\hat{\lambda}=0.53\pm0.66$,
with overall slope $\hat{a}\cdot \hat{\lambda}=0.33 \pm 0.15$).
This chance agreement is in part the result of the fact that
the Poisson data are consistent with the model, as determined by the analysis of the \cstat\ in Sect.~\ref{sec:cstatHypothesis}.
This agreement, however, is not guaranteed in every case, and the standard OLS regression for eteroskedastic data is not recommended
(see, e.g., the discussion in Sect.~3.1 of the textbook by \citealt{cameron2013}).

Finally, it is worth noticing that a weighted least--squares (WLS) linear regression, with weights
equal to the variances of the measurements~\footnote{The weighted least--squares regression
is discussed, for example, in Sect.~19.17 of \cite{kendall1973}.}, is not
recommended in this case. In fact, the WLS estimators depend on the unknown (and unequal) variances,
and therefore it is not possible to determine \emph{a priori} a best--fit WLS model from which then to
infer the Poisson variances for the error analysis, as was done in \eqref{eq:OLSPoissVar}.

\subsubsection{Goodness--of--fit analysis for OLS regression}
Another difference between the maximum--likelihood \cstat\ regression and  the
OLS regression is with regards to goodness--of--fit testing. When
the common variance is unavailable, as is generally the case for the OLS regression,
the test statistic often used is
\begin{equation} t = \dfrac{\hat{b}}{\hat{\sigma}^2_b} \sim t(N-2),
\label{eq:tWald}
\end{equation}
where
\begin{equation}
  \hat{\sigma}^2_b = \dfrac{(1-r^2)}{N-2} \cdot \dfrac{\sum_{i=1}^N(y_i - \overline{y})^2}{\sum_{i=1}^N(x_i - \overline{x})^2}
  \label{eq:VarbR}
\end{equation}
is the sample variance of the slope $b$, and $r$ is the linear correlation coefficient. Notice that this
is a non--parametric estimate of the variance of the slope $b$, in that no assumptions on the
distribution of the data are required.
The linear correlation
coefficient is defined by $r^2=b\, b'$, where $b$ is the usual slope of the linear regression of $Y$ given $X$, and
$b'$ is the slope of the linear regression of $X$ given $Y$ (see, e.g.,
Chapter 14 of \citealt{bonamente2022book}).
This variance differs from that in \eqref{eq:OLSPoissVar}, where the parent distribution
of the measurement is used. From
the observations that $a= \overline{y} - b\, \overline{x}$ and that $(1-r^2) \sum (y_i - \overline{y})^2 / (N-2)$ is the non--parametric
unbiased estimator of the sample variance of $y$,
\eqref{eq:VarbR} yields a non--parametric
estimate of the variance of the parameter $a$ of the regression as
\begin{equation}
  \hat{\sigma}^2_a = \hat{\sigma}^2_b \left( \dfrac{1}{N} \sum_{i=1}^N (x_i - \overline{x})^2 + \overline{x}^2\right),
    \label{eq:VaraR}
\end{equation}
(see, e.g., Sects. 11.5 and 14.3~of \citealt{bonamente2022book}).~\footnote{These
formulas are implemented in the \texttt{scipy.stats.linregress} software.}
The test statistic \eqref{eq:tWald} assumes that the parent slope is null ($b=0$),  it is distributed like a Student's $t$ distribution
with $N-2$ degrees of freedom, and itt is sometimes referred to as the \emph{Wald statistics}, after \cite{wald1943}.
For these data, the linear correlation coefficient is $r^2=0.273$ for a value of $t=1.73$, corresponding to a $p$--value of
$0.12$. Therefore, in the case of the OLS regression, the $t$ statistic provides an alternate non--parametric
means to test the null hypothesis that the
data are uncorrelated, or with $b=0$.

\begin{figure}[t]
  \centering
  \includegraphics[width=5in]{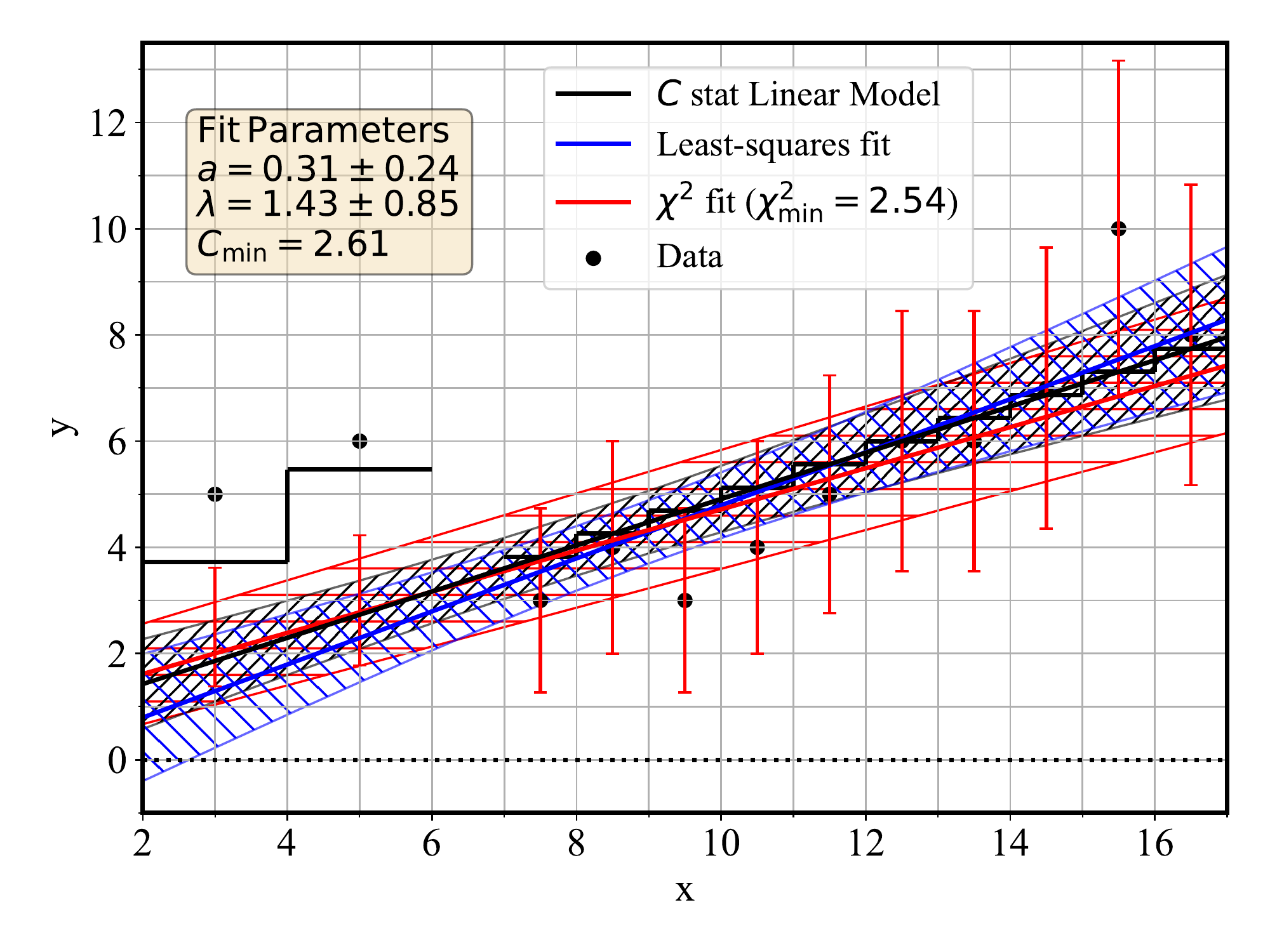}
  \caption{Linear regressions to  days 2--17 of the COVID-19 data of Table~\ref{tab:covid},
  with a gap in the data, and non--uniform binning. The black curves correspond to the
  linear regression with the Poisson--based \cstat, the blue curves for the ordinary least--squares regression, and the red curves to the
  $\chi^2$ regression that assumes heteroskedastic errors reported as red error bars.}
  \label{fig:fit2}
\end{figure}

\subsection{Fit to days 2--16 with a gap in the data, and non--uniform binning}
In this example, the data for days 0 and 1 were ignored, the data for days 2--3 and 4--5 were
combined into bins of size $\Delta x=2$~days, the data for day 6 (where no events were
reported) were ignored, and the datapoints for days 7--16 were kept with the
original binning of $\Delta x=1$~days, for a total
of
$N=12$ measurements. These choices were simply made for illustration purposes.
Specifically, it
is generally not advisable to \emph{ignore} a bin where no counts were recorded,
since the non--detection is actually positive information; a case with no counts--in--bin was
provided in the previous example. Moreover, the linear regression --- and any non--linear regression
as well ---  is sensitive to the choice of bin size. The choice of bin sizes must be dictated by an
understanding of the data at hand, and of the scientific goals of the regression.

\subsubsection{Poisson regression}
The results of the fit to these data are shown in Fig.~\ref{fig:fit2},
for best--fit values of $\hat{a}=0.31\pm0.24$, $\lambda=1.43\pm0.85$,
and an error matrix of
\[ \hat{\varepsilon}=\begin{bmatrix} 0.716 & -0.202 \\
     -0.202 & 0.060
  \end{bmatrix}.
\]

It is  illustrative to compare this estimate to the one from the error propagation method described in Sect.~\ref{sec:errorProp},
which yields an estimated error matrix of
\[ \hat{\varepsilon}_{ep} =
\begin{bmatrix} 0.687 & -0.200 \\
     -0.200 & 0.061
  \end{bmatrix},
\]
where the variance of the measurements $y_k$ were approximated with the
estimated parent means $\hat{\mu}_k$.
This covariance matrix  is nearly identical to the one based on the information
matrix, and the small
differences between $\hat{\varepsilon}$
and $\hat{\varepsilon}_{ep}$
provide an example of the agreement between the two estimates.
Fig.~\ref{fig:fit2} also shows the confidence band for the Poisson regression as the gray hatched area.

\subsubsection{Comparison with OLS and $\chi^2$ regression}
For comparison, a maximum--likelihood regression using the $\chi^2$ statistic
was also performed for this example, and the results are shown as the red line and confidence band.
The ordinary least--squares regression is also shown in blue. For the
method of regression that minimizes the $\chi^2$ statistic,
the data are assumed to be normally--distributed, with variances equal to the number
of counts in each bin. Given the variable bin size, the continuous lines therefore
represent a \emph{density}, i.e., the number of counts per unit interval of the
$x$ variable, which in this case it is time in units of days. This is reflected in the relationship
between the function \eqref{eq:scargle} in units of counts per unit $y$,
and the parent mean in \eqref{eq:mui} in units of pure counts,
which are related by the bin size $\Delta x_i$.
Accordingly, in the $\chi^2$ fit with non--uniform binning, the fit statistic is
\[ \chi^2 = \sum_{i=1}^N \left(\dfrac{y_i - \mu_i}{\hat{\sigma}_i}\right)^2
\]
where $\mu_i$ is given by \eqref{eq:mui} with a bin size $\Delta x_i=2$ for $i=1,2$, and $\Delta x_i=1$
for the other bins.
Equivalently, the $\chi^2$ regression can be performed
directly to the function \eqref{eq:scargle} (or to the linear
model with the standard parameterization $a+b\,x$) by re--scaling both the measurements
$y_i$ and the errors $\hat{\sigma}_i$ to counts per unit $y$, since this is
the usual assumption in standard software that is used for maximum--likelihood
regression for normal data.~\footnote{For example, \texttt{linregress} or \texttt{curve_fit}
in the \texttt{scipy} libraries. Such rescaling of units preserves the
statistical properties of the measurements.} Following this method, the first two data points
have a count rate of one half the values reported by the black dots, and the
$\pm \hat{\sigma}$ error bars for the $\chi^2$ regression are illustrated as red vertical lines;
those error bars have no meaning for the \cstat\ or the OLS regressions.

The overall slope of the \cstat\ regression line is $a \lambda=0.44\pm0.11$, while for the
$\chi^2$ regression it is $0.39\pm0.13$, and for the least--squares fit it is $0.50\pm0.08$.
The differences, albeit relatively small in this application, are the result of
the different assumptions used in the modelling of the data, namely the choice
of Poisson versus normal distribution for the number of counts in each bin.

\section{Discussion and conclusions}
\label{sec:conclusions}
This paper has reviewed the \cite{bonamente2022} maximum--likelihood method of linear regression
for Poisson count data, and has presented an analytical method for the calculation of the covariance matrix
for the parameters of a linear model.
For this method of regression, it is found that the parameterization
of the model of Equation~\ref{eq:scargle}, proposed by \cite{scargle2013},
is quite convenient due to its factorization
of the two parameters,
which results in convenient algebraic properties for the logarithm of the Poisson likelihood.

The method to obtain the covariance matrix
is based on the Fisher information matrix, whose inverse
approximates the covariance matrix in the asymptotic limit of a large number of
measurement. The key result of this paper is that it is possible to obtain
a simple analytical expression for the covariance matrix (Equation~\ref{eq:matrix3})
that can be used for the Poisson linear regression.
This new result, together with the
 maximum--likelihood estimates of the parameters presented in \cite{bonamente2022},
offers a simple and accurate method to perform the linear regression
on a variety of Poisson--distributed count data, including data with non--uniform binning and
with the presence of gaps in the coverage
of the independent variable. This method of regression for Poisson--distributed count
data is therefore available for use in a variety of data applications, such as the one
in the biomedical field presented in Sect.~\ref{sec:application}.

The paper has also shown that an alternative method
for the estimation of the covariance metrix, based on the error propagation or delta method,
also provides a simple analytical estimate that is in general good agreement with
the method based on the information matrix.
The main limitation in the use of the covariance matrix estimated from
the information matrix (as well as that from the error propagation method)
is that it applies only in the limit of a large
number of measurements, where the maximum--likelihood estimates of the parameters
are normally--distributed \citep[e.g.,][]{fisher1922b,cramer1946}.
In this limit, the variance of the parameters and the covariance can be
used for hypothesis testing assuming normal distribution of the parameters.

The maximum--likelihood linear regression to Poisson--distributed count data, presented in this paper
and in \cite{bonamente2022}, have the advantages of being \emph{unbiased} and
\emph{asymptotically efficient}, in the sense that, in the
asymptotic
limit of a large number of measurements, the estimators of the linear model parameters
have the minimum variance bound according to the
\emph{Cram\'{e}r--Rao} theorem (\citealt{rao1945}, \citealt{cramer1946}, see also
\citealt{rao1973} and \citealt{kendall1973}). Given that
maximum--likelihood estimators are also known to be \emph{consistent}
in general (see, e.g.,
Chapter~18 of \citealt{kendall1973}),
this method of regression has very desirable properties for the estimation of the linear model parameters.
On the other hand, the ordinary least--squares regression does retain the
convenient property of unbiasedness, even  for the type of heteroskedastic
variances that apply to these Poisson--distributed data. The OLS, however, is not guaranteed
to be efficient, in the sense that the variances are in general larger than those of the
maximum--likelihood method, and they can also be biased \citep[e.g.][]{kendall1973,swindel1968}.
As a result, the OLS should be viewed as a less accurate method for the regression of Poisson--distributed
count data.

Two other methods of linear regression also discussed in this paper, the weighted least--squares
and the $\chi^2$ methods, suffer from more fundamental shortcomings when
applied to integer--count Poisson data. In both cases, it is necessary to
know \emph{a priori} the unequal variances of the data, in order to proceed with the estimation, and this information
is simply not available. In particular, the $\chi^2$ method is often applied by making the assumption that
the variances are equal to the square root of the counts, \emph{de facto} assuming that
the data are normally distributed and with variances equal to the \emph{measured} counts. Although the
Poisson distribution does converge to a normal distribution in the limit of a large number of counts
(e.g., see Sect.~3 of \citealt{bonamente2022book}), the fact that the variance is approximated with the
measured counts leads to a bias that remains even in the case of large--count data \cite[e.g.][]{kelly2007,bonamente2020}.
It is therefore not appropriate to use either of these two methods for the regression of Poisson--distributed
count data.

\bibliographystyle{imsart-nameyear} 


\end{document}